\documentclass[11pt,a4paper]{article}
\usepackage[margin=25mm]{geometry}
\usepackage{localeng}
\newtheorem{proposition}{Proposition}
\theoremstyle{remark}
\newtheorem*{definition}{Definition}

\usepackage[margin=25mm]{geometry}
\DeclareMathOperator{\KS}{\mathrm{C}\mskip 1 mu}
\DeclareMathOperator{\poly}{\mathrm{poly}\mskip 1 mu}
\newcommand{\cnd}{\mskip 1 mu | \mskip 1 mu}

\title{Individual codewords}
\date{}
\author{Alexander Shen}

\begin{document}

\maketitle

\begin{abstract}
Algorithmic information theory translates statements about classes of objects into statements about individual objects; it defines individual random sequences, effective Hausdorff dimension of individual points, amount of information in individual strings, etc. We observe that a similar translation is possible for list-decodable codes.
\end{abstract}

Let $\mathbb{B}^n$ be the set of $n$-bit strings. The Hamming distance $d(x,y)$ between two strings $x,y\in\mathbb{B}^n$ is the number of coordinates where strings $x$ and $y$  differ.  A set $C\subset \mathbb{B}^n$ is an \emph{$(e,L)$-list-decodable code} if for every $y\in \mathbb{B}^n$ there are at most $L$ strings $x$ such that $x\in C$ and $d(x,y)\le e$. In this definition $e$ and $L$ are two numbers ($e$ is the number of allowed errors and $L$ is the length of the list of possible decoded strings). If we know that some bit string $x$ from $C$ (a \emph{codeword} of the code~$C$)  is transmitted through a channel that changes at most $e$ positions, and see the changed string $y$,  we can find (at most) $L$ possible candidates for $x$.

Being a list decodable code is a property of a code as a set. It makes no sense to ask whether an individual string is a codeword (an element of some list decodable code) or not. Indeed, Boolean cube has many automorphisms: \texttt{xor} with a fixed string preserves Hamming distance and can map any string to any other one. 

It turns out that we can define individual codewords in the framework of the algorithmic information theory.

\begin{definition}
A string $x\in \mathbb{B}^n$ is an \emph{$(e,L)$-codeword} if 
\(
\KS(x\cnd y)\le \log L
\)
for every $y\in\mathbb{B}^n$ such that $d(x,y)\le e$.
\end{definition}
 
Here $\KS(x\cnd y)$ stands for the plain Kolmogorov complexity of $x$ given $y$ (see, e.g.,\cite{suv}). We also use $\KS(x)$ for plain (unconditional) complexity of $x$, and apply $\KS$ to arbitrary finite objects (in particular, codes that are finite sets of strings).
\smallskip

The two-sided connection between this notion and the standard definition of list-decoding codes is provided by the following two (obvious) propositions.

\begin{proposition}\label{prop:atoc}
For every $n$, $e$, and $L$,  the set of all $(e,L)$-codewords is an $(e,O(L))$-list-decodable code. 
\end{proposition}

\begin{proof}
Indeed, for a given $y$ there are at most $2^{\log L+O(1)}=O(L)$ strings such that $\KS(x\cnd y)\le \log L$, and all the $(e,L)$-codewords such that $d(x,y)\le e$, are among them.
\end{proof}

\begin{proposition}\label{prop:ctoa}
Let $C\subset \mathbb{B}^n$ be an $(e,L)$-list-decodable code. Then every element of $C$ is an\\ $(e,L2^{\KS(C)+O(\log n)})$-codeword.
\end{proposition}

\begin{proof}
Consider an arbitrary string $y\in \mathbb{B}^n$. The assumption guarantees that there are at most $L$ strings $x\in C$ such that $d(x,y)\le e$. Each of them can be specified (in addition to condition $y$ and code $C$) by the ordinal number in this list of length $L$, i.e., by $O(\log L)$ bits. Therefore, if $d(x,y)\le e$, the complexity $\KS(x\cnd y)$ is bounded by $\KS(C)+\log L + O(\log \log L)$. The last term $O(\log\log  L)$ is the encoding overhead for the pair $(C,\text{ordinal number})$: recall that $\KS(u,v)\le \KS(u)+\KS(v)+O(\log \KS(v))$. Note also that $O(\log\log L)$ is $O(\log n)$ unless $\log L$ is much bigger than $n$, and the statement is trivial in the latter case. 
\end{proof}

These propositions show that the set of $(e,L)$-codewords is a maximal list-decodable code (that includes every other list-decoding code with some degradation in the parameters, depending on the other code's complexity). Note that if there exists a list decoding code with some parameters, there exists also a \emph{simple} list decoding code with almost the same parameters. This code can be found by searching all subsets; the only thing we have to know are the parameters. If we agree to consider only values of $L$ that are powers of two (losing factor $2$), the complexity of the $(e,L)$-list-decodable code of maximal size is $O(\log n)$. In this way we get the following result:

\begin{proposition}\label{prop:ctoaexist}
Assume that there exists an $(e,L)$-decodable code of cardinality at least $2^k$. Then there exists an $(e, L\cdot\poly(n))$-codeword of complexity at least~$k$.
\end{proposition}

\begin{proof}
Let $L'$ be the maximal power of $2$ that does not exceed $L$. Knowing $\log L'$, $e$, $n$ and $k$ (all of them can be specified by $O(\log n)$ bits), we search for $(e,L')$-list-decodable code of size at least $2^k$. The first code with this property has complexity $O(\log n)$, so Proposition~\ref{prop:ctoa} guarantees that all its codewords are $(e,2^{\log L'+O(\log n)})$-codewords, and $\log L'=\log L+O(1)$. Since there are at least $2^k$ codewords in this code, one of them has complexity at least $k$. 
\end{proof}

The correspondence works in both directions:

\begin{proposition}\label{prop:atocexist}
Assume that there exist an $(e,L)$-codeword of complexity at least $k$. Then there exists an $(e, O(L))$-decodable code of cardinality at least $2^{k-O(\log n)}$.
\end{proposition}

\begin{proof}
The property of being a $(e,L)$-codeword is (computably) enumerable given $e$ and $L$: for a given $x$, we check all the strings $y$ such that $d(x,y)\le e$ and wait until we find that $\KS(x\cnd y)\le \log L$ for all of them. We may again assume without loss of generality that $L$ is a power of $2$ and therefore the parameters have complexity $O(\log n)$. We know that there exist an $(e,L)$-codeword of complexity at least $k$, so there are at least $2^{k-O(\log n)}$ different $(e,L)$-codewords (otherwise each of them could be specified by its ordinal number in the enumeration plus $O(\log n)$ bits for $e$ and $L$). According to Proposition~\ref{prop:atoc}, they form an $(e,O(L))$-list decodable code.
\end{proof}

Propositions~\ref{prop:ctoaexist} and~\ref{prop:atocexist} show that the question ``how large can be an $(e,L)$-decodable code made of $n$-bit strings'' can be reformulated as ``how complex can be an $(e,L)$-codeword of length $n$'', with a small change in parameters.  This does not help much, since the answer to this question is well known (see, e.g., the historical account in~\cite{guruswami} and~\cite{zametki} for the proofs). Moreover, the argument uses random codes and it is not clear whether one can prove the existence of individual codewords of high complexity without going through combinatorial translations. Still this translation gives us one more example of correspondence between combinatorial and complexity languages.

\end{document}